\documentclass[numbook, envcountsect, envcountsame, envcountreset, runningheads, smallextended, natbib, final]{svjour3}     
\smartqed  
\usepackage{graphicx}
 \usepackage{mathptmx}      
\usepackage{geometry}
\usepackage{amsmath}
\usepackage{amssymb}
\usepackage{url} 
\usepackage{marvosym}
\usepackage[colorlinks=false, breaklinks]{hyperref} 
\usepackage{breakurl} 

\newcommand{\euro}{\EUR}
\spnewtheorem{assumption}{Assumption}[section]{\bf}{\it}

 \journalname{Annals of Finance}
\begin{document}

\title{Diversity and Arbitrage in a Regulatory Breakup Model%
}


\author{Winslow Strong         \and
        Jean-Pierre Fouque
}


\institute{Winslow Strong (\Letter) \and Jean-Pierre Fouque \at Department of Statistics and Applied Probability, University of California,
Santa Barbara, CA 93106-3110 \\
              \email{strong@pstat.ucsb.edu}           
\and Jean-Pierre Fouque \at \email{fouque@pstat.ucsb.edu}
}

\date{Received: 21 August 2010 / Accepted: 28 December 2010}


\maketitle

\global\long\global\long\global\long\def\norm#1{\left\Vert #1\right\Vert }

\global\long\global\long\global\long\def\abs#1{\left\vert #1\right\vert }

\global\long\global\long\global\long\def\set#1{\left\{  #1\right\}  }

\global\long\global\long\global\long\def\eps{\varepsilon}

\global\long\global\long\global\long\def\cA{\mathcal{A}}

\global\long\global\long\global\long\def\cB{\mathcal{B}}

\global\long\global\long\global\long\def\C{\mathbb{C}}

\global\long\global\long\global\long\def\bC{\mathbb{C}}

\global\long\global\long\global\long\def\cC{\mathcal{C}}

\global\long\global\long\global\long\def\cD{\mathcal{D}}

\global\long\global\long\global\long\def\cE{\mathcal{E}}

\global\long\global\long\global\long\def\cF{\mathcal{F}}

\global\long\global\long\global\long\def\bF{\mathbb{F}}

\global\long\global\long\global\long\def\cG{\mathcal{G}}

\global\long\global\long\global\long\def\fG{\mathfrak{G}}

\global\long\global\long\global\long\def\bG{\mathbb{G}}

\global\long\global\long\global\long\def\bH{\mathbb{H}}

\global\long\global\long\global\long\def\cH{\mathcal{H}}

\global\long\global\long\global\long\def\cI{\mathcal{I}}

\global\long\global\long\global\long\def\cK{\mathcal{K}}

\global\long\global\long\global\long\def\bL{\mathbb{L}}

\global\long\global\long\global\long\def\cL{\mathcal{L}}

\global\long\global\long\global\long\def\cM{\mathcal{M}}

\global\long\global\long\global\long\def\N{\mathbb{N}}

\global\long\global\long\global\long\def\bN{\mathbb{N}}

\global\long\global\long\global\long\def\cN{\mathcal{N}}

\global\long\global\long\global\long\def\cO{{\normalcolor \mathcal{O}}}

\global\long\global\long\global\long\def\bP{\mathbb{P}}

\global\long\global\long\global\long\def\cP{\mathcal{P}}

\global\long\global\long\global\long\def\bQ{\mathbb{Q}}

\global\long\global\long\global\long\def\cQ{\mathcal{Q}}

\global\long\global\long\global\long\def\R{\mathbb{R}}

\global\long\global\long\global\long\def\bR{\mathbb{R}}

\global\long\global\long\global\long\def\cR{\mathcal{R}}

\global\long\global\long\global\long\def\fR{\mathfrak{R}}

\global\long\global\long\global\long\def\bS{\mathbb{S}}

\global\long\global\long\global\long\def\cS{\mathcal{S}}

\global\long\global\long\global\long\def\cT{\mathcal{T}}

\global\long\global\long\global\long\def\bT{\mathbb{T}}

\global\long\global\long\global\long\def\bU{\mathbb{U}}

\global\long\global\long\global\long\def\bV{\mathbb{V}}

\global\long\global\long\global\long\def\cX{\mathcal{X}}

\global\long\global\long\global\long\def\cY{\mathcal{Y}}

\global\long\global\long\global\long\def\bZ{\mathbb{Z}}

\global\long\global\long\global\long\def\I{\mathbf{1}}

\global\long\global\long\global\long\def\cemetery{\dagger}

\global\long\global\long\global\long\def\D#1#2{\frac{\partial#1}{\partial#2}}

\global\long\global\long\global\long\def\DD#1#2{\frac{\partial^{2}#1}{\partial#2^{2}}}

$\global\long\global\long\global\long\def\vec#1{\mbox{\boldmath\ensuremath{#1}}}
$

\global\long\global\long\global\long\def\wt#1{\widetilde{#1}}

\global\long\global\long\global\long\def\1{\mathbf{1}}

\global\long\global\long\global\long\def\asto{\xrightarrow{\text{a.s.}}}

\global\long\global\long\global\long\def\Lto{\xrightarrow{L^{1}}}

\global\long\global\long\global\long\def\Lpto{\xrightarrow{L^{p}}}

\global\long\global\long\global\long\def\asLto{\xrightarrow{L^{1}, \text{ a.s.}}}

\global\long\global\long\global\long\def\imply{\Rightarrow}

\global\long\global\long\global\long\def\nimply{\nRightarrow}

\global\long\global\long\global\long\def\limply{\Longrightarrow}

\global\long\global\long\global\long\def\leftexp#1#2{{{\vphantom{#2}}^{#1}{#2}}}

\global\long\global\long\global\long\def\essinf{\textrm{ess}\inf}

\global\long\global\long\global\long\def\esssup{\textrm{ess}\sup}

\global\long\global\long\global\long\def\var{\textrm{Var}(}

\begin{abstract}
In 1999 Robert Fernholz observed an inconsistency between the normative
assumption of existence of an equivalent martingale measure (EMM)
and the empirical reality of diversity in equity markets. We explore
a method of imposing diversity on market models by a type of antitrust
regulation that is compatible with EMMs. The regulatory procedure
breaks up companies that become too large, while holding the total
number of companies constant by imposing a simultaneous merge of other
companies. The regulatory events are assumed to have no impact on
portfolio values. As an example, regulation is imposed on a market
model in which diversity is maintained via a log-pole in the drift
of the largest company. The result is the removal of arbitrage opportunities 
from this market while maintaining the market's diversity.
\par \addvspace \medskipamount \noindent \textbf{JEL Classification}\enspace G11
\end{abstract}

\section{Introduction}

What does the empirical phenomenon of diversity in equity markets
imply about investment opportunities in those markets? The answer
depends on the mechanism by which diversity is maintained.

The notion of diversity, the condition that no company's capitalization
(shares multiplied by stock price) may approach that of the entire
market, was introduced by Robert Fernholz in the paper \citet{Art:Fernholz:OnDivEqMark:1999}
and the book \citet{Book:Fernholz:SPT:2002} (see also the recent
review \citet{Art:Karatzas&Fernholz:SPTReview:2009}). He made the
observation in \citet{Art:Fernholz:PortGenFunct:1999} that one of
the most useful tools of financial mathematics, the equivalent martingale
measure (EMM), implies for a large class of models something grossly
inconsistent with real markets: lack of diversity. Historically, the
major world stock markets have been diverse, and they should be expected
to remain so as long as they are subject to a form of antitrust regulation
that prevents concentration of capital into a single company.

Fernholz demonstrated under common assumptions of financial market
modeling that diverse market models necessarily admit strong relative
arbitrage with respect to the market portfolio. Portfolio $\pi$ is
a strong relative arbitrage with respect to portfolio $\rho$ on horizon
$[0,T]$ if $\pi$ strictly outperforms $\rho$ at time $T$ with probability
one. A sufficient set of assumptions are: capitalizations are modeled
by It\^o processes that pay no dividends, trading may occur in continuous
time with no transaction costs, and the covariance process of the
log capitalizations is uniformly elliptic. Importantly, the relative
arbitrage portfolios of Fernholz do not depend on the parameters of
the market, and therefore do not require estimation of these parameters
to construct in practice. They are long-only portfolios (no short
sales) derived from portfolio generating functions (see \citet{Art:Fernholz:PortGenFunct:1999,Book:Fernholz:SPT:2002,Art:FernKaratzKard:DiversityAndRelArb:2005,Art:Karatzas&Fernholz:SPTReview:2009,Art:Fernholz&Karatzas:RelArbVolStab:2005}),
requiring only the weights of the market portfolio as input. If, additionally,
the covariance process is bounded from above uniformly in time, then
no equivalent local martingale measure (ELMM) is possible for such
models. Therefore the fundamental theorem of asset pricing of \citet{Art:DelbSchach:FundThmMathFin:1994}
implies that they admit a free lunch with vanishing risk (FLVR).

To make the case that the argument above pertains to the existence
of (approximate) relative arbitrages in real markets, dividends must
be taken into account. Dividends provide a means for large companies
to slow their growth in terms of capitalization while still generating
competitive total returns (stock return + dividend return) for their
shareholders. An exploratory statistical analysis in \citet{Art:Fernholz:AntitrustNoArb:1998}
of the dividends paid by companies traded on U.S. equity exchanges
from 1967-1996 suggests that this factor has historically been insufficient
to jeopardize the argument for existence of relative arbitrage with
respect to the U.S. market portfolio over this period, before accounting
for transaction costs.

It is not easy to formulate diverse It\^o process models (however
see \citet{Art:OstRhein:ArbOppNonEquivMeasChange:2007} for a clever
probabilistic construction utilizing a non-equivalent measure change).
Almost all market models commonly used in the literature, including
geometric Brownian motion, are not diverse, and therefore do not accurately
model reality. Any diverse It\^o process model with uniformly elliptic
and uniformly bounded covariance must have the characteristic that
the difference in the rate of expected return of the largest company,
compared to some other company, diverges to $-\infty$ as the largest
approaches a relative size cap (see \citet{Art:FernKaratzKard:DiversityAndRelArb:2005}).
Some possible economic rationale to support this type of model includes:
difficulties in achieving high return on investment for very largely
capitalized companies and the cost of antitrust suits brought against
such companies.

Since the onset of antitrust regulation in the U.S. in the late 19th
century, there have been two main regulatory methods of dealing with
companies which get too large: antitrust suits or fines, and antitrust
breakup. The latter is rarely used, with some notable examples being
the breakups of Standard Oil (1911) \citet{Misc:CourtCase:StdOilBreakup1911}
and AT\&T (1982) \citet{Misc:CourtCase:AT&TBreakup1982}. Suits or
fines are used much more often than breakups to discipline companies
that are deemed to be dominating their market in an unfair manner.
Recent examples in Europe include Microsoft in \citet{Misc:CourtCase:AntitrustMicrosoft2004}
in 2004 (\euro497 million) and Intel in \citet{Misc:EUFine:Intel2009}
in 2009 (\euro1.06 billion), both being fined by the European Union
for anticompetitive practices. Models in which diversity is maintained
via the rate of expected return of any company diverging to $-\infty$
as that company's relative size becomes very large can be interpreted
as continuous-path approximations of the case where suits or fines
are used to regulate big companies. Models in which regulatory breakup
is the primary means of maintaining market diversity have not been
well-studied from a mathematical point of view in the financial mathematics
literature. They are the subject of this paper.

When a company is fined money, this directly and adversely affects
the value of the company, so the risk of antitrust fines is a mark
against investing in large companies. In contrast, the key mathematical
feature of a corporate breakup with regards to investment is that
capital need not be removed from the system. That is, when a company
is broken into parts, no net value needs to be lost. Indeed, from
a regulator's perspective, avoidance of monopolies maintains the viability
of an industry's innovation and growth prospects. Although it need
not be the case in practice, for simplicity, we make the modeling
assumption that total market values of companies, as well as the portfolio
values of investors, are \emph{conserved} at each regulatory breakup.
The conservation of portfolio value implies that the capital gains
process from investment in equity is not the stochastic integral of
the trading strategy (shares of equity) with respect to the stock
capitalization process. Instead, a net capitalization process, with
the finite number of regulatory jumps removed, plays the role of integrator.

Another assumption we make is that the number of companies remains
constant. This may seem inconsistent with the breakup of companies,
but in our typical example of regulation we balance the number of
companies in the economy by also requiring that two companies merge
into a new company at the same time as regulation splits a company
into two. This is imposed mainly for mathematical simplicity. It isolates
the effect of regulation on diversity and arbitrage while working
in the familiar context of $\R^{n}$-valued It\^o processes.

As an application we examine a regulated form of a log-pole market
model, a diverse model admitting relative arbitrage with respect to
the market portfolio. The regulation procedure removes the arbitrage
opportunities from the market, resulting in a diverse and arbitrage-free market. Furthermore, the regulated form satisfies the notion
of {}``sufficient intrinsic volatility'' of the market, a more general
sufficient condition for relative arbitrage in unregulated models
(see \citet{Art:Fernholz&Karatzas:RelArbVolStab:2005}). These results
do not contradict the work of Fernholz et al., because in our
model it is the regulated capitalization process that is diverse and
the net capitalization process (which has regulatory jumps removed)
that has an EMM.

This paper is organized as follows. Section \ref{Sec:Premodel} defines
the class of premodels for the regulation procedure, the admissible
trading strategies, portfolios, and the notion of diversity. In Section
\ref{Sec:Overview} we introduce the regulatory procedure, including
defining the regulatory mapping and the triggering mechanism for regulation.
Our exemplar of regulation, the split-merge rule, is also introduced,
which essentially splits the biggest company and forces the smallest
two companies to merge at a regulatory event. The issue of arbitrage
in regulated markets is thoroughly explored and compared to the results
of Fernholz et al. regarding arbitrage and diversity in unregulated
mode. Section \ref{Sec:ExamplesRegMarkets} applies the regulatory
procedure to geometric Brownian motion and to a log-pole market model
to illustrate the compatibility of diversity and EMMs in regulated
models. Section \ref{Sec:Conclusions} presents some concluding remarks
and directions for future research. Section \ref{Sec:Proofs} contains
several proofs.

\section{\label{Sec:Premodel}Premodel}

We first introduce the class of models that we will consider for regulation.
We also define the set of trading strategies that are admissible for
discussions of arbitrage and define the notion of a portfolio for
discussions of relative arbitrage.

The stock capitalization process $\wt X=(\wt X_{1,t},\ldots,\wt X_{n,t})_{t\ge0}^{\prime}$
represents the capitalizations (number of shares multiplied by stock
price) of the $n\ge2$ companies which are traded on an exchange,
where the notation $A^{\prime}$ denotes the transpose of the matrix
$A$. The stock capitalizations are each assumed to be almost surely
(a.s.) strictly positive for all time, with $\wt X$ taking values
in the open, connected, conic set $O^{x}\subseteq\R_{++}^{n}:=(0,\infty)^{n}$.
The dynamics of $\wt X$ is determined by the stochastic differential
equation (SDE)
\begin{align}
d\wt X_{i,t} & =\wt X_{i,t} \left(b_i(\wt X_{t})dt+\sum_{\nu=1}^d\sigma_{i\nu}(\wt X_{t})dW_{\nu,t}\right),\quad 1\le i\le n,\label{Eq:SDE}\\
\wt X_{0} & =x_{0}\in O^{x},\end{align}
for which $(\Omega,\cF,\bF,\wt X,W,P)$  is a solution, where $W$
 is a $d$-dimensional Brownian motion with $d\ge n$. The functions
$b(\cdot)$ and $\sigma(\cdot)$ are assumed to be locally bounded Borel functions.
We require that the SDE (\ref{Eq:SDE}) satisfies strong existence
and pathwise uniqueness for any initial $x_{0}\in O^{x}$, and that
$P(\wt X_{t}\in O^{x},\;\forall t\ge0)=1$.  We shall only consider
volatility matrices $\sigma(x)\in\R^{n\times d}$ having full rank
$n$, $\forall x\in O^{x}$, which guarantees that no stock's risk
can be completely hedged over any time interval by investment in the
other stocks. We assume that  $\cF$ and $\cF_{0}$ contain $\cN$,
the $P$-null sets, and consider only the case where the filtration
is the augmented Brownian filtration  $\bF=\bF^{W}:=\{\cF_{t}^{W}\}_{0\le t<\infty}$,
where $\cF_{t}^{W}:=\sigma\bigl(\bigl\{ W_{s}\bigr\}_{0\le s\le t}\bigr)\bigvee\cN$.

The process $B$ represents a money market account, for which we impose
that a.s. $B\equiv1$, corresponding to zero interest rate. Other
standing assumptions are that capitalizations are exogenously determined,
no dividends are paid, markets are perfectly liquid, trading is frictionless
(no transaction costs) and may occur in arbitrary quantities, and
there are no taxes.

\subsection{Investment in the Premodel}

The model for investment in the risky assets of the premodel is of
the usual type for equity market models. A \emph{trading strategy}
$\wt H_{t}^{\prime}:=(\wt H_{1,t},\ldots,\wt H_{n,t})$ is a predictable
process representing the number of shares held of each stock. Note
that since $\wt X$ is a stock capitalization process, the number
of shares outstanding of each company has effectively been normalized
to one, and so $\wt H$ is with respect to this one share. The \emph{wealth
process }$V^{w,\wt H}$ associated to trading strategy $\wt H$ is
assumed to be \emph{self-financing}, so satisfies \begin{align*}
\wt V_{t}^{w,\wt H} & =\wt H_{t}^{B}+\wt H_{t}^{\prime}\wt X_{t}=w+(\wt H\cdot\wt X)_{t},\end{align*}
where $w$ is the initial wealth and $\wt H^{B}$ is the number of
shares of money market account. We follow Delbaen and Schachermayer's
definition of admissible trading strategies from \citet{Book:DelbSchach:ArbBook:2006}.
\begin{definition}
\emph{\label{Def:AdmTradStrat}Admissible trading strategies }are
predictable processes $\wt H$ such that
\begin{enumerate}
\item $\wt H$ is $\wt X$-integrable, that is, the stochastic integral
$\wt H\cdot\wt X = \big(\int_{0}^{t}\wt H_{s}d\wt X_{s}\big)_{t\ge0}$
is well-defined in the sense of stochastic integration theory for
semimartingales.
\item There is a constant $R$ such that a.s. \begin{align}
(\wt H\cdot\wt X)_{t} & \ge-R,\quad\forall t\ge0.\label{Eq:AdmStratsBB}\end{align}

\end{enumerate}
\end{definition}
The second restriction is designed to rule out {}``doubling strategies''
(see \citet{Book:KS:MathFin:1998}, p.8) and represent the realistic
constraint that credit lines are limited.

It will also be useful in the context of relative arbitrage to develop
the notion of a portfolio, \emph{\`a la} \citet{Art:Karatzas&Fernholz:SPTReview:2009}.
\begin{definition}
\label{Def:Portfolio}A \emph{portfolio }$\wt{\pi}$ is an $\bF$-progressively
measurable $n$-dimensional process bounded uniformly in $(t,\omega)$,
with values in the set\begin{align}
\bigcup_{\kappa\in\N}\bigl\{(\pi_{1},\ldots\pi_{n})\in\R^{n}\mid\pi_{1}^{2}+\ldots+\pi_{n}^{2}\le\kappa^{2},\;\sum_{i=1}^{n}\pi_{i}=1\bigr\}.\label{Eq:PortSet}\end{align}
A \emph{long-only} \emph{portfolio} $\wt{\pi}$ is a portfolio that
takes values in the unit simplex\begin{align*}
\Delta^{n} & :=\bigl\{(\pi_{1},\ldots\pi_{n})\in\R^{n}\mid\pi_{1}\ge0,\ldots\pi_{n}\ge0,\;\sum_{i=1}^{n}\pi_{i}=1\bigr\}.\end{align*}

\end{definition}
A portfolio $\wt{\pi}$ represents the fractional amount of an investor's
wealth invested in each stock. In contrast to a trading strategy,
no borrowing from or lending to the money market is allowed when investment
occurs via a portfolio.  This requirement may be dropped and (\ref{Eq:PortSet})
may be relaxed in favor of more general integrability conditions,
for example see \citet{Art:DFernKaratz:OnOptimalArbitrage:2010}.
However for our purposes here, these restrictions suffice.

For $w\in\R_{++}$, the wealth process $\wt V^{w,\wt{\pi}}$ corresponding
to a portfolio is defined to be the solution to \begin{align}
d\wt V_{t}^{w,\wt{\pi}} & =\wt V_{t}^{w,\wt{\pi}}\sum_{i=1}^{n}\wt{\pi}_{i,t}\frac{d\wt X_{i,t}}{\wt X_{i,t}},\nonumber \\
 & =(\wt V_{t}^{w,\wt{\pi}})\wt{\pi}_{t}^{\prime}\left[b(\wt{X}_t)dt+\sigma (\wt{X}_t ) d\wt W_{t}\right],\label{Eq:PortSDE}\end{align}
which by use of It\^o's formula can be verified to be\begin{align}
\wt V_{t}^{w,\wt{\pi}} & =w\exp\left\{ \int_{0}^{t}\wt{\gamma}_{\wt{\pi},s}ds+\int_{0}^{t}\wt{\pi}_{s}^{\prime}\sigma(\wt{X}_s) d\wt W_{s}\right\} ,\quad\forall t\ge0,\label{Eq:PortExpRep}\end{align}
where \begin{align*}
\wt{\gamma}_{\wt{\pi}} & :=\wt{\pi}^{\prime}b({\wt X})-\frac{1}{2}\wt{\pi}^{\prime} a(\wt{X}) \wt{\pi}\qquad\mbox{and}\qquad a(\cdot):=\sigma(\cdot) \sigma^{\prime} (\cdot).\end{align*}
The process $\wt{\gamma}_{\wt{\pi}}$ is called the \emph{growth rate}
of the portfolio $\wt{\pi}$, and $a(\wt{X})$ is called the \emph{covariance
process}. See \citet{Art:Karatzas&Fernholz:SPTReview:2009} for more
details on the properties of these processes.

The definitions of the wealth process $\wt V^{w,\wt{\pi}}$ corresponding
to a portfolio and $\wt V^{w,\wt H}$ corresponding to a trading strategy
are consistent in the sense that any portfolio has an a.s. unique
corresponding admissible trading strategy yielding the same wealth
process from the same initial wealth. The corresponding trading strategy
$\wt{H}^{w,\wt{\pi}}$ can be obtained from \begin{align}
\wt H_i^{w,\wt{\pi}} & = \frac{\wt{\pi}_i \wt V^{w,\wt{\pi}}}{\wt{X}_i}, \quad 1\le i \le n.\label{Eq:Shares-PortCorrPremodel}\end{align}

The \emph{market portfolio }$\wt{\mu}$ is of particular interest
since {}``beating the market'' is often a desirable goal for investors.
The market portfolio is simply the relative capitalization of each
company in the market with respect to the total: \begin{align*}
\wt{\mu}_{i,t} & :=\mu_{i}(\wt X_{t}):=\frac{\wt X_{i,t}}{\sum_{j=1}^{n}\wt X_{j,t},},\quad1\le i\le n.\end{align*}
The market portfolio is a passive portfolio, meaning that once the initial portfolio $\wt{\mu}$ is setup, it is not traded.  The wealth process $\wt{V}_t^{w,\wt{\mu}}$ is proportional to the total capitalization of the market, as seen by
\begin{align*}
\wt{V}_t^{w,\wt{\mu}} = \Big(\frac{w}{\sum_{j=1}^n \wt{X}_{j,0}}\Big)\sum_{j=1}^n \wt{X}_{j,t},\quad \forall t\ge 0.
\end{align*}
Since the stock capitalization process $\wt X$ a.s. takes values
in $O^{x}\subseteq\R_{++}^{n}$, then for $O^{\mu}:=\mu(O^{x})$,
we have that a.s., $\forall t\ge0$,\begin{align*}
\wt{\mu}_{t}\in O^{\mu}\subseteq\mu(\R_{++}^{n})=\Delta_{+}^{n} & :=\left\{ (\pi_{1},\ldots,\pi_{n})\in\R^{n}\mid\pi_{1}>0,\ldots,\pi_{n}>0, \; \sum_{i}^{n}\pi_{i}=1\right\} .\end{align*}
The closure of a set $A\subseteq\R_{++}^{n}$ will be referred to
as $\bar{A}$ and, unless otherwise stated, is taken with respect
to the subspace topology of $\R_{++}^{n}$, and similarly for subsets
of $\Delta_{+}^{n}$. For example, $\bar{\R}_{++}^{n}=\R_{++}^{n}$
and $\bar{\Delta}_{+}^{n}=\Delta_{+}^{n}$.

\subsection{Diversity}

The notion of diversity entails that no company may ever become too
big in terms of relative capitalization. For generalizations to this
notion and their implications see \citet{Art:FernKaratzKard:DiversityAndRelArb:2005}.
Diversity is a realistic criterion for a market model to satisfy,
since it has held empirically in developed equity markets over time
and should be expected to continue to hold as long as antitrust regulation
prevents capital from concentrating in a single company. In discussions
of diversity it is useful to adopt the reverse-order-statistics notation.
That is, for $x\in\R^{n}$,\begin{align*}
x_{(1)} & \ge x_{(2)}\ge\ldots\ge x_{(n)}.\end{align*}

\begin{definition}
\label{Def:Diversity}A premodel is \emph{diverse }on\emph{ $[0,T]$}
if there exists $\delta\in(0,1)$ such that a.s.\begin{align*}
\wt{\mu}_{(1),t}<1-\delta,\qquad\forall\quad0\le t\le T.\end{align*}
A premodel is\emph{ weakly diverse }on \emph{$[0,T]$ }if there exists
$\delta\in(0,1)$ such that \begin{align*}
\frac{1}{T} & \int_{0}^{T}\wt{\mu}_{(1),t}dt<1-\delta,\quad\mbox{a.s.}\end{align*}
We will not make much use of diversity until later on, but it is good
to keep the definition in mind when considering the regulation procedure proposed herein.
\end{definition}

\section{\label{Sec:Overview}Regulated Market Models}

\subsection{Overview and Modeling Assumptions}

The notion of regulation we introduce consists of confining the market
weights (except at exit times) in an open set $U^{\mu}$ by a regulatory
procedure that
\begin{itemize}
\item conserves the number of companies in the market;
\item conserves total market capital;
\item conserves portfolio wealth;
\item causes a jump in company capitalizations.
\end{itemize}
Upon exit from $U^{\mu}$, the market weights are mapped back into
$U^{\mu}$ by a deterministic mapping $\fR^{\mu}$ applied to $\wt{\mu}$
at its exit point. Then $\mu$ diffuses according to the SDE (\ref{Eq:SDE})
until it exits from $U^{\mu}$ again. The cycle of diffusion and regulation
continues on indefinitely, determining a regulated market weight process
$\mu$. This idea will be made precise in the following subsection.

The economic motivation behind the regulated market models presented
in this paper is to study markets with the feature that companies
may merge and split, possibly forced to do so by a regulator, with
an aim to explore the ramifications for diversity and arbitrage in
these markets. In order to avoid what the authors believe to be unnecessary
mathematical complications in the study of these notions, we require
that splits and merges only occur simultaneously and in pairs, so
that the number of companies in the economy remains a constant. For
example, the biggest company may split into two, and simultaneously
the smallest two merge into one.

Two crucial assumptions of our regulated market models are that total
market capital and portfolio wealth are conserved at each regulation
event. These assumptions are indeed idealizations, but the authors
believe that the former is a reasonable starting point for studying
splits and merges imposed into otherwise continuous path premodels,
while the latter is then sensible in consideration of the below remark.
\begin{remark}
\label{Rem:YHatStochInt}Consider a company being split into two smaller
companies with capitalization fractions $\rho$ and $1-\rho$ relative
to the parent company. If each investor's money in the parent company
is also broken up so that they are left with fraction $\rho$ invested
in the first offspring and $1-\rho$ invested in the second offspring
immediately following the split, then individual portfolio wealth
and total market capital are conserved. This mapping of portfolio
wealth does not impose any constraints on the trading strategies or
portfolios available in the regulated market. Since trading occurs
in continuous time, any investor may simply rearrange all of her money
just after regulation. That our investor may do this without affecting
market prices reflects the assumption that stock capitalizations are
exogenously determined, that is, our investor is small relative to
the market, and her behavior has negligible impact on asset prices.
\end{remark}
The alternative to the wealth conservation assumption would be to
impose a random jump in portfolio wealth at regulation. This would
be compelling for studies of event-driven arbitrage, but since here
our purpose is to study the structural-type arbitrage arising from
diversity, we feel that this is a reasonable omission.

\subsection{\label{Sec:StrongSDE}Regulated Markets}

In this section we will construct the regulated stock process by means
of induction via the diffusion-regulation cycle outlined in the previous
subsection. Since the SDE (\ref{Eq:SDE}) for $\wt X$ satisfies strong
existence and pathwise uniqueness, then we need not pass to a new
probability space to construct the regulated model. Extensions are
possible when (\ref{Eq:SDE}) merely satisfies weak existence and
weak uniqueness, but for simplicity of presentation, we do not pursue
these generalizations here.
\begin{definition}
A \emph{regulation rule} $\fR^{\mu}$ with respect to the open, nonempty
\emph{regulatory set} $U^{\mu}\subseteq O^{\mu}\subseteq\Delta_{+}^{n}$
is a Borel function \begin{align*}
\fR^{\mu} & :\partial U^{\mu}\to U^{\mu}.\end{align*}

\end{definition}
The regulation rule $(U^{\mu},\fR^{\mu})$ uniquely determines the
following set and capital-conserving map of stock capitalizations:
\begin{align*}
U^{x} & :=\mu^{-1}(U^{\mu})\subseteq O^{x},\\
\fR^{x} & :\partial U^{x}\to U^{x},\\
\fR^{x}(x) & :=\left(\sum_{i=1}^{n}x_{i}\right)\fR^{\mu}(\mu(x)).\end{align*}
The inclusion $U^{x}\subseteq O^{x}$ follows from our assumption
that $O^{x}$ is conic, which implies $O^{x}=\mu^{-1}(O^{\mu})$.
 The set $U^{x}$ is conic, that is $x\in U^{x}\imply\lambda x\in U^{x},\;\forall\lambda>0$,
allowing any total market value for a given $\mu\in U^{\mu}$. Therefore,
the  market capitalization $M$ is a degree of freedom for the regulatory
mapping, in the sense that $\mu(\fR^{x}(x))=\mu(\fR^{x}(\lambda x))$,
$\forall\lambda>0$. Specification of $(U^{x},\fR^{x})$ or $(U^{\mu},\fR^{\mu})$
uniquely determines the other, so we refer to either as {}``regulation
rules,'' and in discussion drop the labels and refer to them as $(U,\fR)$.

Define the following processes and random variables: \begin{align*}
W^{1} & :=W, &  &  & X^{1} & :=\wt X,\\
\tau_{0} & :=0, &  &  & \tau_{1} & :=\varsigma_{1}:=\inf\left\{ t>0\mid\mu(X_{t}^{1})\notin U^{\mu}\right\} .\end{align*}
The process $X^{1}$ will serve as the first piece of the regulated
capitalization  process on the stochastic interval $[0,\tau_{1}]$$:=\left\{ (t,\omega)\in[0,\infty)\times\Omega\mid0\le t\le\tau_{1}(\omega)\right\} $.
At $\tau_{1}$, $X^{1}$ has just exited $U^{x}$, so the regulation
procedure maps the capitalization process to $\fR^{x}(X_{\varsigma_{1}}^{1})$,
and the regulated capitalization process continues from that point
according to the dynamics given by the SDE (\ref{Eq:SDE}). To implement
this, define the following variables and processes inductively, $\forall k\in\N$
on $\{\tau_{k-1}<\infty\}$, terminating if $P(\tau_{k-1}<\infty)=0$:
\begin{align}
W_{t}^{k} & :=W_{\tau_{k-1}+t}-W_{\tau_{k-1}},\quad\forall t\ge0,\nonumber \\
dX_{i,t}^{k} & =X_{i,t}^{k} \left(b_{i}(X_{t}^{k})dt+\sum_{\nu=1}^{d}\sigma_{i\nu}(X_{t}^{k})dW_{\nu,t}^{k}\right), \quad 1\le i \le n, \label{Eq:X^kSDE}\\
X_{0}^{k} & =\begin{cases}
y_{0}\in U^{x}, & \mbox{for }k=1,\\
\fR^{x}(X_{\varsigma_{k-1}}^{k-1}), & \mbox{for }k>1,\end{cases}\nonumber \\
\varsigma_{k} & :=\inf\left\{ t>0\mid X_{t}^{k}\notin U^{x}\right\} ,\nonumber \\
\tau_{k} & :=\sum_{j=1}^{k}\varsigma_{j}.\nonumber \end{align}
If for some $k\in\N$ the induction terminates because $P(\tau_{k-1}<\infty)=0$,
then on $\{\tau_{k-1}<\infty\}$, $\forall m\ge k$ define $X^{m}\equiv y_{0}\in U^{x}$,
$\tau_{m}=\infty$, $\varsigma_{m}=0$. Use these same definitions
$\forall m\in\N$ on $\{\tau_{m-1}=\infty\}$. These cases are included
for completeness and their specifics are irrelevant for the subsequent
development.

By the strong Markov property of Brownian motion and stationarity of its increments,
if $P(\tau_{k-1}<\infty)>0$, then for \begin{align*}
\cF_{t}^{k} & :=\cF_{\tau_{k-1}+t},\qquad\mbox{ }\bF^{k}:=\{\cF_{t}^{k}\}_{t\ge0},\end{align*}
 $(W^{k},\bF^{k})$ is a Brownian motion on $\{\tau_{k-1}<\infty\}$,
that is, on $(\Omega\cap\{\tau_{k-1}<\infty\},\cF\cap\{\tau_{k-1}<\infty\})$.
The SDE (\ref{Eq:X^kSDE}) for $k\ge2$ has the same form as the SDE
(\ref{Eq:SDE}) for $\wt X$, but with $W^{k}$ in place of $W$,
and with initial condition $X_{0}^{k}=\fR^{x}(X_{\varsigma_{k-1}}^{k})$
a.s. on $\{\tau_{k-1}<\infty\}$. Therefore, on $\{\tau_{k-1}<\infty\}$
by strong existence, there exists $X^{k}$ adapted to $\bF^{k}$ satisfying
(\ref{Eq:X^kSDE}).

Each $\varsigma_{k}$ is a stopping time with respect to $\bF^{k}$,
since it is the hitting time of the closed set $\R^{n}\setminus U^{x}$
by the continuous process $X^{k}$. Therefore each $\tau_{k}$ is an $\bF$-stopping
time. Since $\fR^{x}(X_{\varsigma_{k-1}}^{k-1})\in U^{x}$, then
$\varsigma_{k}>0$ and $\tau_{k}>\tau_{k-1}$ both a.s. on $\{\tau_{k-1}<\infty\}$,
for all $k$ such that $P(\tau_{k-1}<\infty)>0$. Note that under
this construction there is the possibility of explosion, that is,
of $\tau_{\infty}:=\lim_{k\to\infty}\tau_{k}<\infty$. This will be
considered in greater detail later on.

We are now ready to define the regulated capitalization process $Y$
by pasting together and shifting the $\{X^{k}\}_{1}^{\infty}$ at
the $\{\tau_{k}\}_{1}^{\infty}$ as follows.
\begin{definition}
With respect to regulation rule $(U,\fR)$ and initial point $y_{0}\in U^{x}$,
the \emph{regulated capitalization process} is defined as \begin{align}
Y_{t}(\omega) & :=\begin{cases}
X_{0}^{1}\1_{\{0\}}(t)+\sum_{k=1}^{\infty}\1_{(\tau_{k-1},\tau_{k}]}(t,\omega)X_{t-\tau_{k-1}}^{k}(\omega), & \forall(t,\omega)\in[0,\tau_{\infty}),\\
X_{0}^{1}, & \forall(t,\omega)\notin[0,\tau_{\infty}),\end{cases}\label{Eq:YDef}\end{align}
where $P(X_{0}^{1}=y_{0})=1$. If $P(\tau_{\infty}=\infty)=1$, then
we call the triple $(y_{0},U,\fR)$\emph{ viable} for the premodel.
\end{definition}
To count the number of regulations by time $t$, let $N_{t}:=\sum_{k=1}^{\infty}\1_{\{t>\tau_{k}\}}$,
$t\ge0$. Since each $\tau_{k}$ is a stopping time, $N$
is $\bF$-adapted. Each $X^{k}$ is $\bF^{k}$-progressive,
so by a standard shift argument $Y$ can be seen to be $\bF$-adapted,
and therefore also $\bF$-progressive due to the left-continuity of
its paths.

\subsection{Investment in the Regulated Market}

As remarked earlier, in the regulated model wealth is unaltered by
a regulatory event. Specifically, the wealth process $V^{w,H}$ of
trading strategy $H$ does not jump upon redistribution of market
capital at $\tau_{k}^{+}$. However, the capitalization $Y$ does jump.
This implies that the capital gains of a trading strategy can't be
the stochastic integral of the trading strategy with respect to the
regulated capitalization process. In order to recover the useful tool
of representing the capital gains process as a stochastic integral,
we define a net capitalization process $\widehat{Y}$, which only
accounts for the non-regulatory movements of $Y$.
\begin{definition}
\label{Def:YHat}The \emph{net capitalization process }$\widehat{Y}$
is defined as\begin{align}
\widehat{Y}_{t} & :=\begin{cases}
Y_{t}-\sum_{k=1}^{N_{t}}(\fR^{x}(Y_{\tau_{k}}^{k})-Y_{\tau_{k}}), & \forall(t,\omega)\in[0,\tau_{\infty}),\\
Y_{0}, & \forall(t,\omega)\notin[0,\tau_{\infty}).\end{cases}\label{Eq:YhatDef}\end{align}

\end{definition}
The process $\widehat{Y}$ is $\bF$-adapted since $Y$ and $N$ are
adapted. If the regulated market is viable, then a.s. $\widehat{Y}$
has continuous paths since then a.s. $Y$ has piecewise continuous
paths, jumping only at the $\tau_{k}$. The following representations
of $\widehat{Y}$ will also be useful and are obtainable from the
definitions of $Y$ and $\widehat{Y}$ and (\ref{Eq:X^kSDE}). \begin{align}
\widehat{Y}_{t} & =Y_{0}+\sum_{k=1}^{N_{t}+1}(X_{(t-\tau_{k-1})\wedge\varsigma_{k}}^{k}-X_{0}^{k}),\quad\forall(\omega,t)\in[0,\tau_{\infty})\label{Eq:YhatDeltaX}\\
d\widehat{Y}_{i,t} & =Y_{i,t} \bigl[b_{i}(Y_{t})dt+\sum_{\nu=1}^d\sigma_{i\nu}(Y_{t})dW_{t}\bigr],\quad 1\le i\le n,\quad\mbox{on }[0,\tau_{\infty}).\label{Eq:YhatSDE}\end{align}

The net capitalization process is the correct process to fulfill the
role of integrator for a trading strategy in the regulated market.
To see this, let trading strategy $H$ be the number of shares invested
in the regulated stock process $Y$. The wealth process $V^{H}$ should
be locally self-financing on each stochastic interval $(\tau_{k-1},\tau_{k}]$
and without jumps, so that\begin{align*}
V_{t}^{w,H} & =w+\sum_{k=1}^{N_{t}}\int_{\tau_{k-1}^{+}}^{\tau_{k}}H_{s}dY_{s}+\int_{\tau_{N_{t}}^{+}}^{t}H_{s}dY_{s}.\end{align*}
Then since $\widehat{Y}_{t}-\widehat{Y}_{\tau_{k-1}}=Y_{t}-Y_{\tau_{k-1}}$
on $(\tau_{k-1},\tau_{k}]$ for all $k\in\N$, we have\begin{align*}
V_{t}^{w,H} & =w+\sum_{k=1}^{N_{t}}\int_{\tau_{k-1}}^{\tau_{k}}H_{s}d\widehat{Y}_{s}+\int_{\tau_{N_{t}}}^{t}H_{s}d\widehat{Y}_{s},\\
 & =w+(H\cdot\widehat{Y})_{t}.\end{align*}
This motivates the following natural analog of the usual self-financing
condition.
\begin{definition}
In a viable regulated market, a wealth process $V^{w,H}$ corresponding
to $\widehat{Y}$-integrable trading strategy $H$ is called \emph{self-financing}
in the regulated market if
\begin{align*}
V_{t}^{w,H} & =w+(H\cdot\widehat{Y})_{t},\quad\forall t\ge0.\end{align*}

As in the premodel, in a viable regulated model we will henceforth
assume that all wealth processes are self-financing and that all trading
strategies are $\widehat{Y}$-admissible, which means that $H$ is
$\widehat{Y}$-integrable, and $H\cdot\widehat{Y}$ is a.s. bounded
from below uniformly in time, paralleling Definition \ref{Def:AdmTradStrat}.

\end{definition}
A portfolio in the regulated model will be denoted by $\pi$, and
is a process meeting the requirements of Definition \ref{Def:Portfolio}.
It represents the fractional amount of total wealth invested in the
regulated stocks $Y$.  Paralleling the premodel (\ref{Eq:PortSDE})
for initial wealth $w\in\R_{++}$, the wealth process $V^{w,\pi}$
corresponding to $\pi$ is given by\begin{align}
V_{t}^{w,\pi} & =w\exp\left\{ \int_{0}^{t}\gamma_{\pi,s}ds+\int_{0}^{t}\pi_{s}^{\prime}\sigma (Y_s) dW_{s}\right\} ,\quad\forall t\ge0,\label{Eq:PortExpRepReg}\end{align}
where \begin{align*}
\gamma_{\pi} & :=\pi^{\prime}b(Y)-\frac{1}{2}\pi^{\prime}a(Y)\pi,\qquad a(\cdot)=\sigma(\cdot)\sigma^{\prime}(\cdot).\end{align*}

The market portfolio is the portfolio with the same weights $\mu$
as the market. Note that unlike $\wt H^{w,\wt{\mu}}$, which is constant,
$H^{w,\mu}$ is piecewise constant, jumping at the $\tau_{k}$. All
portfolios, including the market portfolio, have wealth processes
of identical functional form (compare (\ref{Eq:PortExpRepReg}) and
(\ref{Eq:PortExpRep})) in the regulated model and in the premodel.
Therefore, from a mathematical viewpoint, the differences in investment
opportunities in these markets are completely due to the differences
in dynamics of $b(Y)$ and $\sigma(Y)$ compared to $b(\wt X)$ and
$\sigma(\wt X$), which is in turn due to confining $\wt X$ to $U^{x}$
via $\fR^{x}$ to obtain $Y$.

\subsection{\label{SubSec:Split-MergeReg}Split-Merge Regulation}

The exemplar for regulation used in this paper is the split-merge
regulation rule. The basic economic motivation behind split-merge
regulation is that it provides a means for regulators to control the
size of the largest company in the economy. At each $\tau_{k}$, the largest
company is split into two new companies of equal capitalization. In
order to avoid the mathematical complications of a market model with
a variable number of companies, we also impose that at each $\tau_{k}$
the smallest two companies merge, so that the total number of companies
is a constant, $n$. A natural trigger for when regulators might force
a large company to split is company size. For example, regulation
may be triggered when the biggest company reaches $1-\delta$ in relative
capitalization.

The purpose of this subsection is to define the class of split-merge
regulation rules and to find sufficient conditions for the viability
of this class. These results are summarized in Lemma \ref{Lem:baDeltaEntropyPos}.

To identify which company by index occupies the $k$th rank at time
$t$, we use the random function $p_{t}(\cdot)$ so that $\mu_{p_{t}(k),t}=\mu_{(k),t}$,
for $1\le k\le n$. Similarly, for the vector $x:=(x_{1},\ldots,x_{n})$
we use $p(\cdot)$ satisfying $x_{p(k)}=x_{(k)}$, for $1\le k\le n$.
In the event that several components are tied, for example $x_{(k)}=\ldots x_{(k+j)}$,
then ties are settled by $p(k)<\ldots<p(k+j)$.

To define the notion of split-merge regulation, we first define a
regulation prerule, which captures the essential idea but still requires
some technical refinement.
\begin{definition}
In a market where $n\ge3$, a \emph{split-merge} \emph{regulation
prerule} $(U^{\mu},\check{\fR}^{\mu})$ with respect to open, nonempty
regulatory set $U^{\mu}\subseteq O^{\mu}$ is a mapping \begin{align*}
\check{\fR}^{\mu} & :\bar{U}^{\mu}\to\Delta_{+}^{n}\end{align*}
such that \begin{align*}
\check{\fR}^{\mu}(\mu) & =\mu,\quad\forall\mu\in U^{\mu},\end{align*}
and $\check{\fR}^{\mu}\restriction_{\partial U^{\mu}}$ is specified
by the map: \begin{align*}
\mu_{p(1)} & \mapsto\frac{\mu_{p(1)}}{2},\\
\mu_{p(n-1)} & \mapsto\frac{\mu_{p(1)}}{2},\\
\mu_{p(n)} & \mapsto\mu_{p(n)}+\mu_{p(n-1)},\\
\mu_{p(k)} & \mapsto\mu_{p(k)},\quad\forall k:2\le k<n-1.\end{align*}

\end{definition}
The split-merge regulation prerule can be interpreted as splitting
the largest company in half into two new companies and forcing the
smallest two companies to merge into a new company. The condition
$n\ge3$ insures that these companies are distinct. The new companies
from the split are assigned the indices of the previous largest and
the previous second smallest companies. The new company from the merge
is assigned the index of the previous smallest company.
\begin{remark}
Due to the interchange of indices, this interpretation makes economic
sense only in a market model where the companies are taken to be generic,
that is, they have no firm-specific (index-specific) properties. For
example, in a market model where sector-specific correlations are
being modeled, it would not make sense for an oil company resulting
from a split to take over the index of a technology company freed
up from a merge, since the subsequent correlations would not be realistic.
The examples in this paper focus on generic market models, so this
interpretation is sensible for them.
\end{remark}
A split-merge regulation prerule $(U^{\mu},\check{\fR}^{\mu})$ is
not quite suitable for our notion of split-merge regulation, because
in the event that $\mu_{(1),\tau_{k}}=\mu_{(2),\tau_{k}}=\ldots=\mu_{(j),\tau_{k}}$,
we desire that all of these largest companies be broken up, not just
one of them. This can be easily accomplished, however, by repeating
the procedure $n$ times.
\begin{definition}
\label{Def:SplitMergeRule}If $n\ge3$ and split-merge regulation
prerule $(U^{\mu},\check{\fR}^{\mu}$) is into $\bar{U}^{\mu}$, then
we may define\begin{align*}
\fR^{\mu} & :=(\underbrace{\check{\fR}^{\mu}\circ\ldots\circ\check{\fR}^{\mu}}_{n\mbox{ compositions}})\restriction_{\partial U^{\mu}}.\end{align*}
If $\fR^{\mu}$ is into $U^{\mu}$, then we may restrict the codomain
to $U^{\mu}$, and we call the resulting function $(U^{\mu},\fR^{\mu})$
the \emph{split-merge regulation rule }associated with $(U^{\mu},\check{\fR}^{\mu})$.
\end{definition}
Note that the above definition implies that when a split-merge regulation
rule exists, it is a regulation rule. The following technical lemma
will be handy for verifying the viability of split-merge rules. We
use the notation $C_{b}^{2}(\Delta_{+}^{n},\R)$ to denote the continuous
bounded functions from $\Delta_{+}^{n}$ to $\R$ with partial derivatives
continuous and bounded through 2nd order.
\begin{lemma}
\label{Lem:EntropyViability}If the SDE (\ref{Eq:SDE}) has drift
$b(\cdot)$ and volatility $\sigma(\cdot)$ functions which are bounded
on $U^{x}$, and there exists a function $G\in C_{b}^{2}(\Delta_{+}^{n},\R)$
such that the regulation rule $(U,\fR)$ satisfies either \begin{align*}
\inf\left\{ G(\fR^{\mu}(\mu))-G(\mu)\mid\mu\in\partial U^{\mu}\right\}  & >0 \\
\mbox{or} \qquad \sup\left\{ G(\fR^{\mu}(\mu))-G(\mu)\mid\mu\in\partial U^{\mu}\right\} & <0,\end{align*}
where $\partial U^{\mu}$ is the boundary of the set $U^{\mu}$ taken
as a subset of the space $\Delta_{+}^{n}$, then the regulated market
is viable. \end{lemma}
\begin{proof}
See Section \ref{Sec:Proofs}.
\end{proof}
We turn now to the question of identifying suitable regulatory sets
$U$ for split-merge regulation that are both economically compelling
and generate viable split-merge rules.
\begin{lemma}
\label{Lem:baDeltaEntropyPos}Suppose the following hold:
\begin{enumerate}
\item $n\ge3$.
\item $\delta\in(0,\frac{n-1}{n+1})$.
\item The regulatory set, \textup{\begin{align*}
U^{\mu} & :=\{\mu\in\Delta_{+}^{n}\mid\mu_{(1)}<1-\delta\},\end{align*}
}satisfies $U^{\mu}\subseteq O^{\mu}$.
\item $(U^{\mu},\check{\fR}^{\mu}$) is a split-merge regulation prerule.
\item The functions $b(\cdot)$ and $\sigma(\cdot)$ are bounded on $U^{x}$.
\end{enumerate}
Then the split-merge rule $(U^{\mu},\fR^{\mu})$ associated with $(U^{\mu},\check{\fR}^{\mu})$
exists and is viable.

\end{lemma}
\begin{proof}
See Section \ref{Sec:Proofs}.
\end{proof}

\subsection{\label{Sec:Arbitrage}Arbitrage}

We begin with the notions of arbitrage, relative arbitrage, and no
free lunch with vanishing risk (NFLVR). Then we recall the fundamental
theorem of asset pricing (FTAP) for locally bounded semimartingales
and discuss its implications for regulated market models.
\begin{definition}
\label{Def:Arbs}In the premodel an \emph{arbitrage }over $[0,T]$\emph{
}is an admissible trading strategy $\wt H$ such that \begin{align}
P[(\wt H\cdot\wt X)_{T}\ge0]=1\quad\mbox{and}\quad P[(\wt H\cdot\wt X)_{T}>0]>0.\label{Eq:ArbDef}\end{align}
A \emph{relative arbitrage }over\emph{ $[0,T]$ }with respect to portfolio
$\wt{\eta}$ is a portfolio $\wt{\pi}$ such that \begin{align}
P(\wt V_{T}^{1,\wt{\pi}}\ge\wt V_{T}^{1,\wt{\eta}})=1\quad\mbox{and}\quad P(\wt V_{T}^{1,\wt{\pi}}>\wt V_{T}^{1,\wt{\eta}})>0.\label{Eq:RelArbDef}\end{align}
The corresponding notions of \emph{strong arbitrage }and \emph{strong
relative arbitrage} are defined by making the first inequalities of
(\ref{Eq:ArbDef}) and (\ref{Eq:RelArbDef}) strict, respectively.
\end{definition}
The condition NFLVR is a strengthening of the no arbitrage condition,
roughly implying that not only are there no arbitrages, but no {}``approximate
arbitrages.'' See \citet{Book:DelbSchach:ArbBook:2006,Art:DelbSchach:FundThmMathFin:1994,Art:DelbSchach:FundThmMathFin:1998}
for a complete exposition.
\begin{definition}
\label{Def:NFLVR}For $T\in\R_{++}$ define \begin{align*}
\wt K & :=\left\{ (\wt H\cdot\wt X)_{T}\mid\wt H\mbox{ admissible}\right\} ,\end{align*}
which is a convex cone of random variables in $L^{0}(\Omega,\cF_{T},P)$,
and \begin{align*}
\wt C & :=\left\{ \wt g\in L^{\infty}(\cF_{T},P)\mid\wt g\le\wt f\mbox{ for some }\wt f\in\wt K\right\} .\end{align*}
The condition \emph{no free lunch with vanishing risk (NFLVR) }over\emph{
$[0,T]$ }with respect to $\wt X$ is\begin{align*}
\overline{\wt C}\cap L_{+}^{\infty}(\cF_{T},P) & =\{0\},\end{align*}
where $\overline{\wt C}$ denotes the closure of $\wt C$ with respect
to the norm topology of $L^{\infty}(\cF_{T},P)$.
\end{definition}
For the analogs of Definitions \ref{Def:Arbs} and \ref{Def:NFLVR}
in a viable regulated model, simply replace $\wt X$ with $\widehat{Y}$
and remove all other {}``$\;\wt{}\;$''.

The FTAP states that NFLVR for the integrator of the class of trading
strategies is equivalent to the existence of an ELMM for the integrator\emph{.}
In the premodel the integrator is $\tilde{X}$, while in the regulated
model it is $\widehat{Y}$. Since $\widehat{Y}$ obeys the SDE (\ref{Eq:YhatSDE}),
then from standard theory in order for $(\widehat{Y})_{0\le t\le T}$
to be a local martingale under an equivalent measure $Q$ given by
$\frac{dQ}{dP}=:Z_{T}\in \cF_{T}$, then there exists a strictly positive martingale $(Z_t)_{0\le t\le T}$ satisfying $Z_t=E[Z_{T}\mid\cF_{t}]$ and having the representation:\begin{align}
Z_{t} & :=\cE(-\theta(Y)\cdot W)_{t}=\exp\left\{ -\int_{0}^{t}\theta(Y_{s})^{\prime}dW_{s}-\frac{1}{2}\int_{0}^{t}\vert\theta(Y_{s})^{\prime}\vert^{2}ds\right\} ,\quad0\le t\le T,\label{Eq:SDEViableZExpMart}\end{align}
where the market price of risk $\theta(\cdot)$ solves the market
price of risk equation\begin{align}
\sigma(Y_{t})\theta(Y_{t})=b(Y_{t}),\quad \text{a.s., } \; 0\le t\le T.\label{Eq:MarkPriceRiskEq}\end{align}
Given an exponential local martingale of the form (\ref{Eq:SDEViableZExpMart}), satisfaction of the Novikov criterion,\begin{align}
E\biggl[\exp\biggl\{\frac{1}{2}\int_{0}^{T}\left|\theta_{s}\right|^{2}ds\biggr\}\biggr] & <\infty,\label{Eq:Novikov}\end{align}
is sufficient for implying the martingality of $(Z)_{0\le t\le T}$.

While the NFLVR condition is of theoretical interest, it is not necessarily
of practical relevance. If we put ourselves in the situation of having
to select from some set of candidate market models, some of which
satisfy NFLVR and others of which do not, it may be a hopeless task
to figure out whether financial data support or refute NFLVR. In fact,
example 4.7 \citet{Art:KaratzKard:Numeraire:2007} shows that two
general semimartingale models on the same stochastic basis may possess
the same triple of predictable characteristics, with one admitting
an arbitrage while the other does not. Even if we have reason to believe
that a model admitting arbitrage or relative arbitrage is an accurate
one, it may be the case that the arbitrage portfolios depend in a
delicate way on the parameters of the model, $b$ and $\sigma$ here.
In such a case any attempts to estimate these parameters from observed
data would likely be too imprecise to lead to an investment strategy
that could convincingly be called an approximation to an arbitrage.

In contrast to this, the condition of diversity is supported by world
market data and the existence of antitrust laws in developed markets.
The condition (\ref{Eq:UniformEllipticity}) of uniform ellipticity
of the covariance is not as readily apparent, but seems to be a reasonable
manifestation of the idea that there is always at least some baseline
level of volatility in markets. The significance of these two conditions
is that in unregulated market models together they imply the existence
of a long-only relative arbitrage portfolio that is functionally generated
from the market weights (see Section \ref{Subsect:Diversity,-Intrinsic-Volatility,}
herein for the precise conditions as well as \citet{Art:Fernholz:PortGenFunct:1999,Art:Fernholz&Karatzas:RelArbVolStab:2005,Art:Karatzas&Fernholz:SPTReview:2009}),
not requiring estimation of $b$ or $\sigma$. It is therefore of
great interest whether or not this implication carries over to regulated
markets. The following proposition will be useful in Section \ref{Sec:ExamplesRegMarkets}
for showing that this is not the case.
\begin{proposition}
\label{Prop:BddVolNoRelArbREG}If the regulated model is viable, $\widehat{Y}$
satisfies NFLVR over $[0,T]$, and $\sigma(\cdot)$ is bounded on
$U^{x}$, then any ELMM for $\widehat{Y}$ is an EMM, and no portfolio
is a relative arbitrage with respect to any other portfolio over $[0,T]$
in the regulated model.\end{proposition}
\begin{proof}
To prove the martingality, let measure $Q$ be any ELMM for $\widehat{Y}$, and therefore have the form $\frac{dQ}{dP}=Z_{T}=\cE(-\theta(Y) \cdot W)_{T}$, where $\theta (Y)$ solves the market price of risk equation (\ref{Eq:MarkPriceRiskEq}).
Then Equation (\ref{Eq:YhatSDE}) implies that\begin{align*}d\widehat{Y}_{i,t} & =Y_{i,t}(\sum_{\nu=1}^d\sigma_{i\nu} (Y_t) dW_{\nu,t}^{(Q)}), \quad 1\le i\le n, \quad0\le t\le T,\end{align*}
where\begin{align*}
W_{t}^{(Q)} & :=W_{t}+\int_{0}^{t}\theta (Y_s) ds,\quad0\le t\le T\end{align*}
is a $Q$-Brownian motion by Girsanov's theorem. This implies
that \begin{align*}
dV_{t}^{w,\pi} & =V_{t}^{w,\pi}\pi_{t}^{\prime}\sigma (Y_t)dW_t^{(Q)},\quad0\le t\le T.\end{align*}
Therefore, $(V^{w,\pi})_{0\le t\le T}$ is an exponential $Q$-local martingale.
Since $Y\in U^{x}$, $dt\times dP$-a.e., this implies that $\sigma(Y)$
is bounded, $dt\times dP$-a.e. The portfolio $\pi$ is uniformly
bounded by definition, so $(V_{t}^{w,\pi})_{0\le t\le T}$ is a $Q$-martingale
by the Novikov criterion (\ref{Eq:Novikov}).

Now suppose that $\pi$ is a relative arbitrage with respect to $\eta$.
Then by $Q\sim P$ it follows that \begin{align*}
Q(V_{T}^{w,\pi}\ge V_{T}^{w,\eta})=1\quad\mbox{and}\quad Q(V_{T}^{w,\pi}>V_{T}^{w,\eta})>0.\end{align*}
However $(V_{t}^{w,\pi})_{0\le t\le T}$ and $(V_{t}^{w,\eta})_{0\le t\le T}$
are both $Q$-martingales, so their difference is also a $Q$-martingale,
with $E^{Q}[V_{T}^{w,\pi}-V_{T}^{w,\eta}]=w-w=0$. This contradicts
the relative arbitrage property above, so this market admits no pair
of relative arbitrage portfolios.\qed\end{proof}
Some recent investigations pertaining to relative arbitrage include
\citet{Art:FernKaratzKard:DiversityAndRelArb:2005}, \citet{Art:Fernholz&Karatzas:RelArbVolStab:2005},
\citet{Art:DFernholz&Banner:ShortTermRelArbVolStab:2008}, \citet{Art:Ruf:OptTradStratUnderArb:2009:PREPRINT},
and \citet{Art:Mijatovic&Urusov:DetermCritAbsOfArbInDiffModels:2009:PREPRINT},
to name a few. An arbitrage is essentially a relative arbitrage with
respect to the money market account, modulo the uniform boundedness
requirement of portfolios and their prohibition from investing in
the money market, both of which can be relaxed as in \citet{Art:DFernKaratz:OnOptimalArbitrage:2010}.
The existence of a relative arbitrage does not imply the existence
of an arbitrage, as illustrated by examples, often called ``bubble
markets'' (see \citet{Art:CoxHobson:LocalMartsBubblesOptPrices:2005,Art:PalProtterBubblesHTransform:2010,Art:JarrowProtterShimbo:BubblesComplete:2007}),
where there exists an equivalent measure under which the stock process
is a strict local martingale. In particular, if $\wt \pi$ is a relative
arbitrage with respect to $\wt \eta$, then the trading strategy $\wt{H}:=\wt{H}^{1,\wt \pi}-\wt{H}^{1,\wt \eta}$
need not satisfy the requirement that $\wt{H} \cdot\wt X$ be uniformly
bounded from below, so $\wt H$ need not be admissible.

\subsection{\label{Subsect:Diversity,-Intrinsic-Volatility,}Diversity, Intrinsic
Volatility, and Relative Arbitrage}

The works by Robert Fernholz et al. (\citet{Art:Fernholz:OnDivEqMark:1999,Book:Fernholz:SPT:2002,Art:Fernholz&Karatzas:RelArbVolStab:2005})
on diversity and arbitrage prove that for unregulated markets, over
an arbitrary time horizon, there exist strong relative arbitrage portfolios
with respect to the market portfolio in any weakly diverse market
satisfying certain assumptions and regularity conditions. Furthermore,
they show how such relative arbitrages can be constructed as long-only
portfolios which are functionally generated from $\wt{\mu}$, not
requiring knowledge of $\wt b$ or $\wt{\sigma}$. A sufficient set
of assumptions and regularity are given by the following.
\begin{assumption}
\label{Ass:SuffCondRelArb}
\begin{enumerate}
\item \label{Item:Ass:ItoProc}The capitalizations are modeled by an It\^o
process\begin{align*}
d\wt X_{i,t} & =\wt X_{i,t}\left(\wt{b}_{i,t} dt+ \sum_{\nu=1}^d \wt{\sigma}_{i\nu,t}dW_{\nu,t}\right),\quad 1\le i\le n,\\
\wt{X}_{0} & =x_{0}\in\R_{++}^{n},\end{align*}
where $\wt b$ and $\wt{\sigma}$ are progressively measurable processes
satisfying $\forall T\in\R_{++}$, \[
\sum_{i=1}^{n}\left(\int_{0}^{T}|\wt b_{i,t}|dt+\sum_{\nu=1}^{d}\int_{0}^{T}\bigl|\wt{\sigma}_{i\nu,t}\bigr|^{2}dt\right)<\infty,\quad\mbox{a.s.}\]

\item \label{Item:Ass:UEllptic}The capitalizations' covariance process
is uniformly elliptic:\begin{align}
\exists\varepsilon>0:\mbox{ a.s.}\quad & \varepsilon\left|\xi\right|^{2}\le\xi\wt{\sigma}_{t}\wt{\sigma}_{t}^{\prime}\xi,\quad\forall t\ge0,\;\forall\xi\in\R^{n}.\label{Eq:UniformEllipticity}\end{align}

\item \label{Item:Ass:NoDivs}Companies pay no dividends (and therefore
can't control their size by this means).
\item \label{Item:ConstantCompanies}The number of companies is a constant.
\item \label{Item:Ass:WeakDiv}The market is weakly diverse.
\item \label{Item:Ass:TradingRules}Trading may occur in continuous time,
in arbitrary quantities, is frictionless, and does not impact prices.
\end{enumerate}
\end{assumption}
These conditions have been generalized in \citet{Art:Fernholz&Karatzas:RelArbVolStab:2005}.
There it is shown that the uniform ellipticity assumption may be relaxed,
and the market need not be weakly diverse if it satisfies one of several
notions of {}``sufficient intrinsic volatility.'' One measure of
the intrinsic volatility in the market is the excess growth rate of
the market portfolio, \begin{align*}
\gamma_{\wt{\mu},t}^{*} & =\frac{1}{2}\left(\sum_{i=1}^{n}\wt{\mu}_{i,t}\wt a_{ii,t}-\wt{\mu}_{t}^{\prime}\wt a_{t}\wt{\mu}_{t}\right).\end{align*}
The following proposition provides an example of a {}``sufficient
intrinsic volatility'' type condition.
\begin{proposition}[adapted from Proposition 3.1 \citet{Art:Fernholz&Karatzas:RelArbVolStab:2005}]
\label{Prop:SuffIntrVol} Assume an unregulated market model satisfies
items \ref{Item:Ass:ItoProc}, \ref{Item:Ass:NoDivs}, \ref{Item:ConstantCompanies},
and \ref{Item:Ass:TradingRules} of Assumption \ref{Ass:SuffCondRelArb}.
Additionally suppose there exists a continuous, strictly increasing
function $\wt{\Gamma}:[0,\infty)\to[0,\infty)$ with $\wt{\Gamma}(0)=0$,
$\wt{\Gamma}(\infty)=\infty$, and satisfying a.s. \begin{align}
\wt{\Gamma}(t) & \le\int_{0}^{t}\wt{\gamma}_{\wt{\mu},s}^{*}ds<\infty,\qquad\mbox{for all}\quad0\le t<\infty.\label{Eq:SuffIntrVol}\end{align}
Then there exists a functionally generated, long-only portfolio that
is a strong relative arbitrage with respect to the market portfolio
over sufficiently long horizon.\end{proposition}
\begin{proof}
See \citet{Art:Fernholz&Karatzas:RelArbVolStab:2005}.
\end{proof}
A diverse regulated market is simply a regulated market in which $\mu$
in place of $\wt{\mu}$ satisfies Definition \ref{Def:Diversity}.
By Lemma 3.4 of \citet{Art:Karatzas&Fernholz:SPTReview:2009} (the
proof of which is merely algebraic and has nothing to do with whether
the market model is regulated or not) in a uniformly elliptic, diverse
market (respectively, regulated market), $\wt{\gamma}_{\wt{\mu}}^{*}$
($\gamma_{\mu}^{*}$) satisfies\begin{align}
\frac{\varepsilon\delta}{2} & \le\wt{\gamma}_{\wt{\mu},t}^{*},\quad\forall t\ge0\qquad\left(\frac{\varepsilon\delta}{2}\le\gamma_{\mu,t}^{*},\quad\forall t\ge0\right).\label{Eq:ExcGrowthBB}\end{align}
In this equation $\varepsilon$ satisfies $\varepsilon\left|\xi\right|^{2}\le\xi\wt{\sigma}_{t}\wt{\sigma}_{t}^{\prime}\xi$
($\varepsilon\left|\xi\right|^{2}\le\xi\sigma_{t}\sigma_{t}^{\prime}\xi$),
$\forall t\ge0,\;\forall\xi\in\R^{n}$, and $\delta$ satisfies $\mu_{(1)}\le1-\delta$
($\wt{\mu}_{(1)}\le1-\delta$). This implies that in any uniformly
elliptic, diverse market (regulated market), that (\ref{Eq:SuffIntrVol})
(its regulated market counterpart) is satisfied by $\wt{\Gamma}(t)=\frac{\varepsilon\delta}{2}t$
($\Gamma(t)=\frac{\varepsilon\delta}{2}t$). In the examples of Section
\ref{Sec:ExamplesRegMarkets}, NFLVR and no relative arbitrage hold
for the regulated markets, while diversity and uniform ellipticity
also hold, implying that (\ref{Eq:SuffIntrVol}) is satisfied in these
cases. \emph{Therefore, in contrast to the premodel, the conditions of weak
diversity and uniform ellipticity together, and thus also the weaker condition of sufficient
intrinsic volatility, are not sufficient for the existence of relative
arbitrage in the regulated model.
}

\section{\label{Sec:ExamplesRegMarkets}Examples of Regulated Markets}

In this section we apply the split-merge regulation of subsection
\ref{SubSec:Split-MergeReg} to geometric Brownian motion (GBM) and
a log-pole market as premodels. In both cases the regulated market
is diverse and uniformly elliptic, and therefore satisfies the regulated
market analog of the sufficient intrinsic volatility condition (\ref{Eq:SuffIntrVol}).
In both cases the regulated market satisfies NFLVR and admits no pair
of relative arbitrage portfolios.

\subsection{Geometric Brownian Motion}

Consider the case where the unregulated capitalization process is
a GBM,\begin{align*}
d\wt X_{i,t} & =\wt X_{i,t} \left[b_i dt+ \sum_{\nu=1}^n \sigma_{i\nu} dW_{\nu, t}\right], \quad 1\le i\le n\\
\wt X_{0} & =x_{0}\in O^{x}=\R_{++}^{n},\end{align*}
 for some $n\ge3$, $b\in\R^{n}$, and $\sigma\in\R^{n\times n}$
of rank $n$. GBM satisfies NFLVR on all $[0,T]$, $T\in\R_{++}$
and has constant volatility, so it is not weakly diverse on any $[0,T]$
and admits no pair of relative arbitrage portfolios (see Section 6
of \citet{Art:Karatzas&Fernholz:SPTReview:2009}). Select $\delta\in(0,\frac{n-1}{n+1})$
and define the regulatory set\begin{align*}
U^{\mu} & :=\{\mu\in\Delta_{+}^{n}\mid\mu_{(1)}<1-\delta\}.\end{align*}
By Lemma \ref{Lem:baDeltaEntropyPos} the associated split-merge rule
exists and is viable. Since $\theta:=\sigma^{-1}b$ is a constant,
the Novikov criterion (\ref{Eq:Novikov}) for $Z:=\cE(-\theta\cdot W)$
is satisfied, and $Z$ is therefore a martingale. This implies that
for any $T\in\R_{++}$, $Q$ specified by $\frac{dQ}{dP}:=Z_{T}$
is an ELMM for $(\widehat{Y}_{t})_{0\le t\le T}$. Furthermore, $(\widehat{Y}_{t})_{0\le t\le T}$
is a $Q$-martingale, and the regulated market is free of relative
arbitrage by Proposition \ref{Prop:BddVolNoRelArbREG}. The regulated
market is diverse since $P(\mu_{t}\in\bar{U}^{\mu},\;\forall t\ge0)$$=P(\mu_{(1),t}\le1-\delta,\;\forall t\ge0)=1$,
which implies that (\ref{Eq:ExcGrowthBB}) and thus (\ref{Eq:SuffIntrVol})
are satisfied. Therefore in this regulated market, the notions of
sufficient intrinsic volatility and diversity coexist with NFLVR and
no relative arbitrage.

\subsection{Log-Pole Market\label{SubSec:Log-Pole}}

So-called {}``log-pole'' market models provide examples of diverse,
unregulated markets. Diversity is maintained in these markets by means
of a log-pole-type singularity in the drift of the largest capitalization,
diverging to $-\infty$ as the largest weight $\mu_{(1)}$ approaches
the diversity cap $1-\delta$. Explicit portfolios which are relative
arbitrages with respect to the market portfolio over any prespecified
time horizon may be formed by down-weighting the largest company in
a controlled manner (see e.g. \citet{Art:FernKaratzKard:DiversityAndRelArb:2005,Art:Karatzas&Fernholz:SPTReview:2009}).
This model can be interpreted as a continuous approximation of an
economy in which the relative size of the largest company is controlled
via a regulator imposing fines on it. When regulatory breakup is applied,
keeping the largest weight $\mu_{(1)}$ away from $1-\delta$, then
the arbitrage opportunities vanish.

Following Section 9 of \citet{Art:Karatzas&Fernholz:SPTReview:2009}
(see \citet{Art:FernKaratzKard:DiversityAndRelArb:2005} for more
details and generality) fix $n\ge3$, $\delta\in(0,\frac{1}{2})$
and consider the unregulated capitalization process $\wt X$, the
pathwise unique strong solution to \begin{align*}
d\wt X_{i,t} & =\wt X_{i,t} \left(b_i (\wt X_{t})dt+\sum_{\nu=1}^n\sigma_{i\nu} W_{\nu,t}\right), \quad 1\le i\le n,\\
\wt X_{0} & =x_{0}\in O^{x}:=\{x_{0}\in\R_{++}^{n}\mid\mu_{(1)}(x_{0})<1-\delta\},\end{align*}
where $\sigma\in\R^{n\times n}$ is rank $n$. The function $b(\cdot)$
is given by \begin{align*}
b_{i}(x) & :=\frac{1}{2}a_{ii}+g_{i}1_{\cQ_{i}^{c}}(x)-\frac{c}{\delta}\frac{1_{\cQ_{i}}(x)}{\log\left(\left(1-\delta\right)/\mu_{i}(x)\right)},\qquad1\le i\le n,\end{align*}
where $\{g_{i}\}_{1}^{n}$ are non-negative numbers, $c$ is a positive
number, and \begin{align*}
\cQ_{1} & :=\left\{ x\in\R_{++}^{n}\mid x_{1}\ge\max_{2\le j\le n}x_{j}\right\} ,\qquad\cQ_{n}:=\left\{ x\in\R_{++}^{n}\mid x_{n}>\max_{1\le j\le m-1}x_{j}\right\} ,\\
\cQ_{i} & :=\left\{ x\in\R_{++}^{n}\mid x_{i}>\max_{1\le j\le i-1}x_{j},\quad x_{i}\ge\max_{i+1,\le j\le n}x_{j}\right\} ,\quad\mbox{for }i=2,\ldots,n-1.\end{align*}
When $x\in\cQ_{i}$, then $x_{i}$ is the largest of the $\{x_{j}\}_{1}^{n}$
with ties going to the smaller index. In this model each company behaves
like a geometric Brownian motion when it is not the largest. The largest
company is repulsed away from the log-pole-type singularity in its
drift at $1-\delta$. Strong existence and pathwise uniqueness for
this SDE are guaranteed for any $x_{0}$ in $O^{x}$ by \citet{Art:Veretennikov:1981}
(see also \citet{Art:FernKaratzKard:DiversityAndRelArb:2005}). The
capitalizations satisfy $P(\wt X_{t}\in O^{x},\; \forall t\ge0)=1$,
so this premodel is diverse. The function $b(\cdot)$ is locally bounded
since the coefficients of $1_{\cQ_{i}^{c}}(x)$ and $1_{\cQ_{i}}(x)$
are continuous on $O^{x}$ and the singularity at $\mu_{(1)}(x)=1-\delta$
is away from the boundary of each $\cQ_{i}$ for $\delta\in(0,\frac{1}{2})$.
Since the market is diverse and has constant volatility, then by
the results of \citet{Book:Fernholz:SPT:2002,Art:Karatzas&Fernholz:SPTReview:2009}
over arbitrary horizon the market admits long-only relative arbitrage
portfolios which are functionally generated from the market portfolio.
Furthermore since $\sigma$ is a constant, $\wt X$ has no ELMM so
admits a FLVR.

This model may be regulated in such a way to remove these relative
arbitrage opportunities and satisfy NFLVR.  Picking $\delta^{\prime}\in(\delta,\frac{n-1}{n+1})$
and $x_{0}\in U^{x}$, define the regulatory set to be \begin{align*}
U^{\mu} & :=\{\mu\in\Delta_{+}^{n}\mid\mu_{(1)}<1-\delta^{\prime}\}\subseteq O^{\mu}.\end{align*}
The associated split-merge regulation rule exists and is viable by
Lemma \ref{Lem:baDeltaEntropyPos}. The function $b\restriction_{\bar{U}^{x}}(\cdot)$
is bounded, so taking $\theta(\cdot):=\sigma^{-1}b(\cdot)$, then
$\theta(Y)$ is a.s. bounded uniformly in time. This implies
that the Novikov criterion (\ref{Eq:Novikov}) is satisfied for $Z:=\cE(-\theta (Y) \cdot X)$,
and so $Z$ is a martingale. For any $T\in\R_{++}$, $Q$ specified
by $\frac{dQ}{dP}:=Z_{T}$ is an ELMM for $(\widehat{Y}_{t})_{0\le t\le T}$.
Furthermore $(\widehat{Y}_{t})_{0\le t\le T}$ is a $Q$-martingale,
and the regulated market is free of relative arbitrage by Proposition
\ref{Prop:BddVolNoRelArbREG}. The diversity of the regulated market
implies that (\ref{Eq:ExcGrowthBB}), and thus (\ref{Eq:SuffIntrVol})
are satisfied. Therefore in this regulated market, the notions of
sufficient intrinsic volatility and diversity coexist with NFLVR and
no relative arbitrage.

The pathology of this premodel is that the largest company's drift
approaches $-\infty$ as $\mu_{(1)}$ approaches $1-\delta$. The
cure is to prevent the largest company from approaching $1-\delta$
by regulation and thus bound the worst expected rate of return. The
pathological region of $\Delta_{+}^{n}$ is removed from $\mu$'s
state space by the regulation procedure, and the result is an arbitrage-free
market.

\section{\label{Sec:Conclusions}Conclusions}

Models in which diversity is maintained by a drift-type condition,
whereby the rate of expected return of the largest company must become
unboundedly negative compared to the rate of expected return of some
other company in the economy, cover only one particular mechanism
by which diversity may be achieved. These are reasonable models for
markets in which diversity is maintained by some combination of fines
on big companies imposed by antitrust regulators, and/or the biggest
company consistently delivering less return than the other companies
for other reasons. In such markets there is an intuitive undesirability
in holding the stock of the largest company, since its upside potential
is limited relative to that of the other companies. Fernholz showed
that this is not merely a vague undesirability, but that any passive
portfolio holding shares of the biggest company can be strictly outperformed
by functionally generated portfolios which are relative arbitrages
with respect to the former.

If regulators maintain diversity within an equity market by utilizing
regulatory breakup, then the situation is quite different. This mechanism
need not open the door to arbitrage. It entails no systematic debasement
of the total capital in the economy, and, for many models, can be shown
to be arbitrage-free, admitting an equivalent martingale measure.

The current situation in U.S. markets is that regulatory breakups
are uncommonly used, and primarily in cases reversing provisionally
approved mergers. This suggests that the previous conclusion of Fernholz,
Karatzas et al. in \citet{Art:Fernholz&Karatzas:RelArbVolStab:2005}
that in the past conditions in U.S. markets have likely been compatible
with functionally generated relative arbitrage with respect to the
market portfolio, is not threatened by this result. If, however, regulatory
breakup were to become a primary tool of antitrust regulators, then,
modulo our assumption of portfolio wealth conservation, the argument
for existence of functionally generated relative arbitrage in diverse
markets would be substantially weakened.

The notions of diversity combined with uniform ellipticity, and the more general {}``sufficient intrinsic volatility
of the market'' are useful conditions in that they can be tested by empirical observations. This is in contrast to the rather
abstract and normative condition of existence of an equivalent martingale
measure, for which it may be hopeless to make a case for or against
via observed data alone. That these conditions do not imply relative
arbitrage in regulated market models prompts the question of whether
a general, empirically verifiable condition can be found that implies
relative arbitrage for both regulated and unregulated market models.

\section{\label{Sec:Proofs}Proofs}
\begin{proof}[Proof of Lemma \ref{Lem:EntropyViability}]
For $\mu_{t}:=\mu(Y_{t})$, let $G_{t}:=G(\mu_{t}),\;\forall(\omega,t)\in[0,\tau_{\infty})$.
By Definition \ref{Def:YHat} of $\widehat{Y}$ and \ref{Eq:YhatSDE},
we can decompose $G_{t\wedge\tau_{k}}$ as \begin{align}
G_{t\wedge\tau_{k}} & =G_{0}+\sum_{m=1}^{k}\int_{t\wedge\tau_{m-1}^{+}}^{t\wedge\tau_{m}}dG_{t}+\sum_{m=1}^{N_{t}\wedge(k-1)}\left[G(\fR^{\mu}(\mu_{\tau_{m}}))-G_{\tau_{m}}\right].\label{Eq:Proof:DeltaG:GDecomp1}\end{align}
On $(\tau_{k-1},\tau_{k}]$ by It\^o's formula, the process $\mu$
obeys\begin{align*}
d\mu_{i,t} & =\mu_{i,t} \Big[\Big((b_i(X_t)-\sum_{j=1}^n a_{ij}(X_t) \mu_{j,t}-[\mu_{t}^{\prime}b(X_t)-\mu_{t}^{\prime}a(X_t)\mu_{t}]\Big)dt\\
 & \quad +\sum_{\nu=1}^d\left(\sigma_{i\nu}(X_t)-[\mu_{t}^{\prime}\sigma (X_t)]_\nu\right)dW_{\nu,t}\Big],\\
  & =B_{i,t}dt+\sum_{\nu=1}^dR_{i\nu,t}dW_{\nu,t},\quad 1\le i\le n.\end{align*}
The processes
$B$ and $R$ are bounded on $(0,\tau_{\infty})$, since $b(\cdot)$
and $\sigma(\cdot)$ are uniformly bounded on $U^{x}$. Defining $\widehat{G}_{t}:=G_{t}-\sum_{m=1}^{N_{t}}\left[G(\fR^{\mu}(\mu_{\tau_{m}}))-G_{\tau_{m}}\right]$,
$\forall(t,\omega)\in[0,\tau_{\infty})$, then by It\^o's formula $\widehat{G}$
is an It\^o process on $[0,\tau_{\infty})$, and so there exist processes
$C$ and $S$ taking values in $\R^{n}$ and $\R^{n\times d}$, respectively,
such that\begin{align*}
d\widehat{G}_{t} & =C_{t}dt+S_{t}dW_{t},\quad\mbox{on }(0,\tau_{\infty}).\end{align*}
The integrands $C$ and $S$ are uniformly bounded on $(0,\tau_{\infty})$
since the first and second derivatives of $G(\cdot)$ are by assumption
bounded on $\Delta_{+}^{n}$, and $B$, $R$ above are uniformly bounded
on $(0,\tau_{\infty})$. This implies that $\int_{0}^{t\wedge\tau_{\infty}}C_{s}ds$
and $\int_{0}^{t\wedge\tau_{\infty}}S_{s}dW_{s}$ are well-defined
for all $t>0$ by the theories of Lebesgue and stochastic integration.
Therefore $\lim_{k\to\infty}(\1_{\{\tau_{\infty}<\infty\}}\widehat{G}_{\tau_{k}})\in\R$
a.s.

By (\ref{Eq:Proof:DeltaG:GDecomp1}) and the definition of $\widehat{G}$,
we have: \begin{align}
G_{\tau_{k}} & =\widehat{G}_{\tau_{k}}+\sum_{m=1}^{k-1}\left[G(\fR^{\mu}(\mu_{\tau_{m}}))-G_{\tau_{m}}\right].\label{Eq:Proof:DeltaG:GDecomp2}\end{align}
On $\{\tau_{\infty}<\infty\}$ by assumption either \begin{align*}
\sum_{m=1}^{k-1}\left[G(\fR^{\mu}(\mu_{\tau_{m}}))-G_{\tau_{m}}\right] & \underset{k\to\infty}{\longrightarrow}\infty,\quad\mbox{a.s.},\end{align*}
or \begin{align*}
\sum_{m=1}^{k-1}\left[G(\fR^{\mu}(\mu_{\tau_{m}}))-G_{\tau_{m}}\right] & \underset{k\to\infty}{\longrightarrow}-\infty,\quad\mbox{a.s.}\end{align*}
But in (\ref{Eq:Proof:DeltaG:GDecomp2}) $\{\widehat{G}_{\tau_{k}}\}_{1}^{\infty}$
converges in $\R$ a.s. on $\{\tau_{\infty}<\infty\}$, and $G(\cdot)$
is a bounded function by assumption, so (\ref{Eq:Proof:DeltaG:GDecomp2})
implies that $P(\tau_{\infty}<\infty)=0$.
\qed\end{proof}

\begin{proof}[Proof of Lemma \ref{Lem:baDeltaEntropyPos}]
Fix $n\ge3$ and $\delta\in(0,\frac{n-1}{n+1})$. The boundary of
$U^{\mu}$ in $\Delta_{+}^{n}$ is \begin{align*}
\partial U^{\mu} & =\{\mu\in\Delta_{+}^{n}\mid\mu_{(1)}=1-\delta\}.\end{align*}
The set $U^{\mu}$ is non-empty and open, and by assumption satisfies
$U^{\mu}\subseteq O^{\mu}$. To check that $\check{\fR}^{\mu}$ is
into $\bar{U}^{\mu}$, note that $\mu\in\partial U^{\mu}\imply$$\mu_{(1)}=1-\delta$
$\imply\mu_{(n)}+\mu_{(n-1)}\le\frac{2\delta}{n-1}<\frac{2}{n+1}<1-\delta$,
where the first inequality follows from $\sum_{j=1}^{n}\mu_{j}-\mu_{(1)}=\delta$,
implying that the smallest two weights can sum to at most $\frac{2\delta}{n-1}$,
and the second inequality follows from $\delta\in(0,\frac{n-1}{n+1})$.
This implies that all of the {}``new companies'' created by $\check{\fR}$
are of relative size strictly smaller than $1-\delta$. So $[\check{\fR}(\mu)]_{(1)}\le1-\delta$
which implies that $\check{\fR}$ is into $\bar{U}^{\mu}$. If there
were $k$ companies of relative size $1-\delta$ for $\mu\in\partial U^{\mu}$,
then $\check{\fR}^{\mu}(\mu)$ has $k-1$ companies of relative size
$1-\delta$. Therefore, applying the $n$-fold composition $(\check{\fR}^{\mu}\circ\ldots\circ\check{\fR}^{\mu})$
to $\mu\in\bar{U}^{\mu}$ results in no companies of relative size
$1-\delta$. This implies that $\fR^{\mu}$ of Definition \ref{Def:SplitMergeRule}
is into $U^{\mu}$, making $(U^{\mu},\fR^{\mu})$ a regulation rule
and therefore a split-merge rule.

Consider the entropy function\begin{align*}
S:\R_{++}^{n} & \to\R,\\
S(x) & =-\sum_{i=1}^{n}x_{i}\log x_{i}.\end{align*}
We examine the change in entropy resulting from $\check{\fR}$. For
$\mu\in\partial U^{\mu}$ we have $\mu_{(1)}=1-\delta$, and so\begin{align*}
S(\check{\fR}^{\mu}(\mu))-S(\mu) & =-\left[2\frac{\mu_{(1)}}{2}\log(\frac{\mu_{(1)}}{2})+(\mu_{(n)}+\mu_{(n-1)})\log(\mu_{(n)}+\mu_{(n-1)})\right]\\
 & \;\;\;\;+\left[\mu_{(1)}\log\mu_{(1)}+\mu_{(n)}\log\mu_{(n)}+\mu_{(n-1)}\log\mu_{(n-1)}\right],\\
 & =(1-\delta)\log2-(\mu_{(n)}+\mu_{(n-1)})\log(\mu_{(n)}+\mu_{(n-1)})\\
 & \;\;\;\;+2\left(\frac{\mu_{(n)}\log\mu_{(n)}+\mu_{(n-1)}\log\mu_{(n-1)}}{2}\right).\end{align*}
Applying Jensen's inequality to the convex function $x\mapsto x\log x$,
we get \begin{align*}
S(\check{\fR}^{\mu}(\mu))-S(\mu) & \ge(1-\delta)\log2\\
&\quad + (\mu_{(n)}+\mu_{(n-1)})\left[-\log(\mu_{(n)}+\mu_{(n-1)})+\log(\frac{\mu_{(n)}+\mu_{(n-1)}}{2})\right],\\
 & =(1-\delta)\log2-(\mu_{(n)}+\mu_{(n-1)})\log2,\\
 & \ge\log2\left[1-\delta-\frac{2\delta}{n-1}\right]>0,\end{align*}
where the second to last inequality follows from the fact that $\sum_{j=1}^{n}\mu_{j}-\mu_{(1)}=\delta$,
so the smallest two weights can sum to at most $\frac{2\delta}{n-1}$.
The last inequality follows from the supposition that $\delta\in(0,\frac{n-1}{n+1})$.
From this, the change in entropy of $\fR$ can be seen to satisfy\begin{align*}
S(\fR^{\mu}(\mu))-S(\mu) & \ge\left[1-\delta\left(\frac{n+1}{n-1}\right)\right]\log2>0,\quad\forall\mu\in\partial U^{\mu}.\end{align*}
For $\varepsilon\in\R_{++}$, we may define the shifted entropy function
\begin{align*}
S^{(\varepsilon)} & :\Delta_{+}^{n}\to\R\\
S^{(\varepsilon)}(\mu) & :=S(\varepsilon\1_{n}+\mu)=-\sum_{i=1}^{n}(\mu_{i}+\varepsilon)\log(\mu_{i}+\varepsilon),\end{align*}
where $\1_{n}$ is the column vector of $n$ ones. For any $\kappa\in(0,\infty)$,
the entropy function $S$ restricted to domain $\{\mu:0<\mu_{i}\le\kappa,\mbox{ for }1\le i\le n\}$
is uniformly continuous, so therefore $\varepsilon\in(0,1)$ can be
chosen such that \begin{align}
\inf\left\{ S^{(\varepsilon)}(\fR^{\mu}(\mu))-S^{(\varepsilon)}(\mu)\mid\mu\in\partial U^{\mu}\right\}  & >0.\label{Proof:DeltaEntropy:InfSupShiftedEntropy}\end{align}
The shifted entropy function satisfies $S^{(\varepsilon)}\in C_{b}^{2}(\Delta_{+}^{n},\R)$,
so for an SDE (\ref{Eq:SDE}) with $b(\cdot)$ and $\sigma(\cdot)$
bounded on $U^{x}$, an application of Lemma \ref{Lem:EntropyViability}
with $G=S^{(\varepsilon)}$ proves the viability of $(U,\fR)$.
\qed\end{proof}

\begin{acknowledgements}
The authors would like to thank Ioannis
Karatzas for contributing to the beginning of this research and for
suggestions and references based on a draft, Tomoyuki Ichiba for giving
useful feedback on early drafts, and Pamela Shisler, J.D., for providing
information on U.S. antitrust law. The work of Fouque was partially
supported by National Science Foundation grant DMS-0806461.
\end{acknowledgements}

\bibliographystyle{spmpscinat} 

\end{document}